\pdfoutput=1

\documentclass[a4paper,twocolumn]{svjour3}

\usepackage{amsmath}
\usepackage{flushend}

\usepackage{siunitx}
\usepackage{url}
\usepackage{algorithm}
\usepackage{graphicx}
\graphicspath{{./figures/}}
\usepackage{subfig}
\usepackage{framed}
\usepackage{multirow}
\usepackage{algorithmic}
\newcommand{\degree}{\mathop{\mathrm{degree}}}
\newcommand{\microop}{$\mu$op}
\newcommand{\microops}{\microop{}s}

\usepackage{todonotes}
\newcommand\given[1][]{\:#1\lvert\:}
\usepackage[utf8x]{inputenc}
\usepackage{booktabs}

\usepackage{listings}

\usepackage{times}

\begin{document}


\title{Faster 64-bit universal hashing using carry-less multiplications}

\author{Daniel Lemire         \and
        Owen Kaser 
}

\institute{D. Lemire \at LICEF Research Center, TELUQ, Universit\'e du Qu\'ebec \\
Montreal, QC Canada \\
\email{lemire@gmail.com}
\and
O. Kaser \at 
 Dept.\ of CSAS, University of New Brunswick, Saint John\\
  Saint John, NB Canada\\
\email{o.kaser@computer.org}
}

\date{}

\maketitle

\begin{abstract}
Intel and AMD support the Carry-less Multiplication (CLMUL) instruction set in their  x64 processors. 
We use CLMUL to implement an almost universal 64-bit hash family (\textsc{CLHASH}). 
We compare this new family with what might be the fastest
almost universal family  on x64 processors (\textsc{VHASH}). We find that \textsc{CLHASH} is at least  \SI{60}{\percent} faster. 
We also compare \textsc{CLHASH} with a popular hash function designed for speed (Google's CityHash). We find that \textsc{CLHASH} is \SI{40}{\percent} faster than CityHash on inputs larger than 64~bytes and just as fast otherwise.
\end{abstract}




%
%
%

\section{Introduction}

Hashing is the fundamental operation of mapping data objects to fixed-size hash values. For example, all objects in the Java programming language can be hashed to 32-bit integers. 
Many algorithms and data structures rely on hashing: e.g., authentication codes, Bloom filters and hash tables. We typically assume that given two data objects, the probability that they have the same hash value (called a \emph{collision}) is low.
When this assumption fails, adversaries can negatively impact the performance of these 
data structures or even create denial-of-service attacks. To mitigate such problems, we can pick hash functions at random (henceforth called \emph{random hashing}).

Random hashing is standard in Ruby, Python and Perl. It is allowed explicitly in Java and C++11. 
There are many fast random hash families --- e.g., MurmurHash, Google's CityHash~\cite{cityhash}, SipHash~\cite{siphash2012} and \textsc{VHASH}~\cite{dai2007vhash}.
Cryptographers have also   designed fast hash families with strong theoretical guarantees~\cite{Bernstein2005,MMH1997,krovetz2007message}. However, much of this work predates the introduction of the CLMUL instruction set in commodity x86 processors. Intel and AMD added CLMUL and its \texttt{pclmulqdq} instruction to their processors to accelerate some common
cryptographic operations.
Although the \texttt{pcl\-mul\-qdq} instruction first became available
in 2010, its high cost in terms of CPU cycles --- specifically an 8-cycle throughput on pre-Haswell Intel microarchitectures and a 7-cycle throughput on pre-Jaguar AMD microarchitectures --- limited its usefulness outside of cryptography.
  However, the throughput of the instruction on the newer Haswell architecture is down to 2~cycles, even though it remains a high latency operation (7 cycles)~\cite{fog2014instruction,intelintrin}.\footnote{The  low-power AMD Jaguar microarchitecture does even better with a throughput of 1~cycle and a latency of 3~cycles. 
}  See Table~\ref{ref:simdinstructions}. 
  Our main contribution is to show that the \texttt{pclmulqdq} instruction can be used to produce a 64-bit string hash family that is faster
than known approaches
 while offering stronger theoretical guarantees.

\begin{table*}
\caption{\label{ref:simdinstructions}Relevant SIMD intrinsics and instructions on  \emph{Haswell} Intel processors, with latency and reciprocal throughput in CPU cycles per instruction~\cite{fog2014instruction,intelintrin}. 
 }\centering
\begin{tabular}{lcp{1.6in}cc}
\toprule
intrinsic & instruction & description & latency & rec.\ thr. 
\\\midrule
\_mm\_clmulepi64\_si128 & \texttt{pclmulqdq} & 64-bit carry-less multiplication & 7 &2 \\
\_mm\_or\_si128 & \texttt{por} & bitwise OR & 1 & 0.33\\ 
\_mm\_xor\_si128 &\texttt{pxor} & bitwise XOR & 1 & 0.33\\ 
\_mm\_slli\_epi64 & \texttt{psllq} & shift left two 64-bit integers & 1 & 1 \\
\_mm\_srli\_si128 & \texttt{psrldq} & shift right by $x$ bytes & 1 & 0.5\\
\_mm\_shuffle\_epi8 & \texttt{pshufb} & shuffle 16 bytes & 1 & 0.5\\
\_mm\_cvtsi64\_si128 & \texttt{movq} & 64-bit integer as  128-bit reg. & 1 & --\\
\_mm\_cvtsi128\_si64 & \texttt{movq} & 64-bit integer from 128-bit reg. & 2 & --\\
\_mm\_load\_si128 & \texttt{movdqa} & load a 128-bit reg.\ from memory (aligned) & 1 & 0.5\\
\_mm\_lddqu\_si128 & \texttt{lddqu} & load a 128-bit reg.\ from memory (unaligned) & 1 & 0.5\\
\_mm\_setr\_epi8 & -- & construct 128-bit reg.\ from 16~bytes & -- & --\\
\_mm\_set\_epi64x & -- & construct 128-bit reg.\ from two 64-bit integers & -- & --\\
\bottomrule
\end{tabular}
\end{table*}

\section{Random Hashing}

In random hashing, we pick a hash function at random from some family, 
whereas an adversary might pick the data inputs.
We want distinct  objects to be unlikely to hash to the same value. That is, we want a low collision probability.

We consider hash functions from $X$ to $[0,2^L)$.  An $L$-bit family  is 
  \emph{universal}~\cite{carter1979universal,Cormen:2009:IAT:1614191} if the probability of a collision is no more than $2^{-L}$. That is, it is universal if
 \begin{align*}P\left (h(x)=h(x')\right )\leq 2^{-L}\end{align*}
for any fixed $x,x' \in X$ such that $x \neq x'$, given that we pick $h$ at random from the family.
It is \emph{$\epsilon$-almost universal}~\cite{188765} (also written $\epsilon$-AU) if 
the probability of a collision is bounded by $\epsilon$.  I.e., 
 $P\left (h(x)=h(x')\right )\leq \epsilon$,
for any $x,x' \in X$ such that $x \neq x'$. (See Table~\ref{tab:notation}.)

\begin{table}
\caption{\label{tab:notation} Notation and basic definitions }
\centering
\begin{tabular}{l|p{4cm}} \toprule
$h: X \to \{0,1,\ldots, 2^L-1\}$  & $L$-bit hash function\\[5pt]
universal & $P\left (h(x)=h(x')\right )\leq 1/2^{L}$
for $x \neq x'$\\[5pt]
$\epsilon$-almost universal & $P\left (h(x)=h(x')\right )\leq \epsilon$
for $x \neq x'$\\[5pt]%
XOR-universal & $P\left (h(x) = h(x') \oplus  c \right ) \leq 1/2^{L}$ for any $c \in [0,2^L)$ and distinct $x,x' \in X$ \\[5pt] 
$\epsilon$-almost XOR-universal & $P\left (h(x) = h(x') \oplus  c   \right ) \leq \epsilon$ for any integer $c \in [0,2^L)$ and distinct $x,x' \in X$ \\
\bottomrule
\end{tabular}
\end{table}

\subsection{Safely Reducing Hash Values}

Almost universality can be insufficient to prevent frequent collisions since a given algorithm might only use the first few bits of the hash values. Consider hash tables. A hash table might use as a key only the first $b$~bits of the hash values when its capacity is $2^b$. Yet even if a hash family is $\epsilon$-almost universal, it could still have a high collision probability on the first few bits.

For example, take any $32$-bit universal family $\mathcal{H}$, and derive the new $64$-bit $1/2^{32}$-almost universal $64$-bit family 
by taking the functions from $\mathcal{H}$ and multiplying them by $2^{32}$:
 $h'(x) = h(x) \times  2^{32}$. Clearly, all functions from this new family collide with probability 1 on the first 32~bits, even though the collision probability on the full hash values is low ($1/2^{32}$). 
Using the first bits of
these hash functions could have disastrous consequences in the implementation of a hash table.

Therefore, 
we consider stronger forms of universality.
\begin{itemize}
\item A family is  \emph{$\Delta$-universal}~\cite{stinson1996connections,squarehash} if 
\begin{align*}P(h(x) = h(x') + c \bmod 2^{L}) \leq 2^{-L}\end{align*} for any constant $c$
and any $x,x' \in X$ such that $x \neq x'$. It is  $\epsilon$-almost $\Delta$-universal  if
$P(h(x) = h(x') + c  \bmod  2^{L} \leq \epsilon$ for any constant $c$ and any $x,x' \in X$ such that $x \neq x'$. 
\item A family  is 
$\epsilon$-almost XOR-universal  if
\begin{align*}P\left (h(x) = h(x') \oplus  c   \right ) \leq \epsilon\end{align*} for any integer constant $c\in [0,2^L)$ and any $x,x' \in X$ such that $x \neq x'$ (where $\oplus$ is the bitwise XOR). A family that is 
$1/2^{L}$-almost XOR-universal is said to be  XOR-universal~\cite{stinson1996connections}.
\end{itemize}

Given an $\epsilon$-almost $\Delta$-universal family $\mathcal{H}$ of hash functions  $h:X \to [0,2^L)$, the family of hash functions
\begin{align*}\{ h(x) \bmod 2^{L'} \given  h \in \mathcal{H}\}\end{align*} from $X$ to $[0,2^{L'})$ is 
$2^{L-L'} \times \epsilon$-almost $\Delta$-universal~\cite{dai2007vhash}. The next lemma shows that a similar result applies to almost XOR-almost universal families.

\begin{lemma}\label{lemma:xorismarvellous}
Given an $\epsilon$-almost XOR-universal family $\mathcal{H}$ of hash functions  $h:X \to [0,2^L)$ and any positive integer $L'<L$, the family of hash functions
$\{ h(x) \bmod 2^{L'} \given  h \in \mathcal{H}\}$ from $X$ to $[0,2^{L'})$ is 
$2^{L-L'} \times  \epsilon$-almost XOR-universal.
\end{lemma}
\begin{proof}
For any integer constant $c\in [0,2^L)$, consider the equation $h(x)  = (h(x') \oplus c )\bmod 2^{L'}$ for $x \neq x'$ with $h$ picked from  $\mathcal{H}$. Pick any positive integer $L'<L$. We have \begin{align*}P( &h(x)    = (h(x') \oplus c  \mod 2^{L'}) ) \\
&   =   \sum_{z \given z \bmod 2^{L'}=0 } P(h(x) = h(x')\oplus c \oplus z )\end{align*} where the sum is over $2^{L-L'}$ distinct $z$ values. Because $\mathcal{H}$ is $\epsilon$-almost XOR-universal, we have that
$P(h(x) = h(x')\oplus c \oplus z ) \leq \epsilon$ for any $c$ and any $z$.
Thus, we have that $P( h(x)    = h(x') \oplus c  \mod 2^{L'} ) \leq 2^{L-L'} \epsilon$, showing the result.
\end{proof}

It follows from Lemma~\ref{lemma:xorismarvellous} that if a family is XOR-universal, then its modular reductions are XOR-universal as well.

As a straightforward extension of this lemma, we could show that 
when picking any $L'$~bits (not only
the least significant), the result is $2^{L-L'} \times  \epsilon$-almost XOR-universal.

\subsection{Composition}

It can be useful to combine different hash families to create new ones. For example, it is common to compose hash families.
When composing hash functions ($h=g \circ f$), the universality degrades linearly: if $g$ is picked from an $\epsilon_g$-almost universal family and
$f$ is picked (independently) from 
an $\epsilon_f$-almost universal family, the result is $\epsilon_g+\epsilon_f$-almost universal~\cite{188765}.

We sketch the proof. For $x\neq x'$, we have that $g(f(x))=g(f(x'))$ collides if $f(x)=f'(x)$. This occurs with probability at most $\epsilon_f $ since $f$ is picked from an $\epsilon_f$-almost universal family. If not, they collide if $g(y)=g(y')$ where $y=f(x)$ and $y'=f(x')$, with probability bounded by $\epsilon_g$. Thus, we have bounded the collision probability by $\epsilon_f +(1-\epsilon_f) \epsilon_g \leq \epsilon_f+\epsilon_g$, establishing the result.

By extension, we can show that if $g$ is picked from an $\epsilon_g$-almost XOR-universal family, then the composed result ($h=g \circ f$) is going to be $\epsilon_g+\epsilon_f$-almost XOR-universal. It is not required for $f$ to be almost XOR-universal.

\subsection{Hashing Tuples}\label{sec:hashingtuples}

If we have universal hash functions from $X$ to $[0,2^L)$, then we can construct hash functions from $X^m$ to $[0,2^L)^m$ while preserving universality. The construction is straightforward: $h'(x_1, x_2,\ldots, x_m)= (h(x_1), h(x_2),\ldots, h(x_m))$. If $h$ is picked from  an $\epsilon$-almost universal family, then the result is $\epsilon$-almost universal. This is true even though a single $h$ is picked and reused $m$~times.

\begin{lemma}\label{lemma:hashingtuples}Consider an $\epsilon$-almost universal family $\mathcal{H}$ from $X$ to $[0,2^L)$. Then 
consider
the family of functions 
$\mathcal{H}'$ 
of the form $h'(x_1, x_2,\ldots, x_m)= (h(x_1), h(x_2),\ldots, h(x_m))$ from $X^m$ to $[0,2^L)^m$,  where $h$ is in  $\mathcal{H}$.
Family $\mathcal{H}'$
 is  $\epsilon$-almost universal.
\end{lemma}
The proof is not difficult. 
Consider two distinct values from $X^m$, $x_1, x_2,\ldots, x_m$ and $x'_1, x'_2,\ldots, x'_m$. Because the tuples are distinct, they must differ in at least one component: $x_i \neq x'_i$. It follows that 
$h'(x_1, x_2,\ldots, x_m)$ and $h'(x'_1, x'_2,\ldots, x'_m)$
collide with probability at most $P(h(x_i)=h(x'_i))\leq \epsilon$, showing the result.

\subsection{Variable-Length Hashing From Fixed-Length Hashing}\label{sec:var}

Suppose that  we are given a family  $\mathcal{H}$  of hash functions that is XOR universal over fixed-length strings. That is, we have that
$P\left (h(s) = h(s') \oplus  c \right ) \leq 1/2^{L}$ if the length of $s$ is the same as the length of $s'$ ($|s|=|s'|$). We can create a new family that is XOR universal
over variable-length strings by introducing a hash family on string lengths. Let $\mathcal{G}$ be a family of XOR universal hash functions $g$ over length values. Consider the new family of hash functions of the form $h(s) \oplus  g (|s|)$ where $h \in \mathcal{H}$ and $g \in \mathcal{G}$. Let us consider two distinct strings $s$ and $s'$. There are two cases to consider.
\begin{itemize}
\item If $s$ and $s'$ have the same length so that $g (|s|)=g (|s'|)$ then we have XOR universality since \begin{align*}&P\left (h(s) \oplus  g (|s|) = h(s')\oplus  g (|s'|) \oplus  c \right ) \\&= P\left (h(s)  = h(s')\oplus  c \right )\\ &\leq 1/2^{L}\end{align*} where the last inequality follows because $h\in \mathcal{H}$, an XOR universal family over fixed-length strings.
\item If the strings have different lengths ($|s|\neq|s'|$), then we again have XOR universality because
\begin{align*}&P\left (h(s) \oplus  g (|s|) = h(s')\oplus  g (|s'|) \oplus  c \right )\\ &= P\left (g (|s|)  = g (|s'|)\oplus ( c \oplus h(s) \oplus h(s'))\right )\\ 
&= P\left (g (|s|)  = g (|s'|)\oplus c'\right )\\
&\leq 1/2^{L}\end{align*}
where we set $c'=c \oplus h(s) \oplus h(s')$, a value independent from $|s|$ and $|s'|$. The last inequality follows because $g$ is taken from a family $\mathcal{G}$ that is XOR universal.
\end{itemize}
Thus the result ($h(s) \oplus  g (|s|)$) is XOR universal. We can also generalize the analysis. Indeed, if $\mathcal{H}$ and $\mathcal{G}$ are $\epsilon$-almost universal, we could show that the result is $\epsilon$-almost universal. We have the following lemma.

\begin{lemma}\label{lemma:fromfixedtovariable}
Let  $\mathcal{H}$  be an XOR universal family of hash functions  over fixed-length strings. Let $\mathcal{G}$ be an XOR universal family of  hash functions over integer values. We have that the family of hash functions of the form $s\to h(s) \oplus  g (|s|)$ where $h \in \mathcal{H}$ and $g \in \mathcal{G}$ is XOR universal over all strings. 

Moreover, if $\mathcal{H}$ and $\mathcal{G}$ are merely $\epsilon$-almost universal, then  the family of hash functions of the form $s\to h(s) \oplus  g (|s|)$ is also $\epsilon$-almost universal. 
\end{lemma}

\subsection{Minimally Randomized Hashing}
Many
 hashing algorithms --- for instance,
CityHash~\cite{cityhash} ---
 rely on a small random \emph{seed}. The 64-bit version of CityHash takes a 64-bit integer as a seed. 
Thus, we effectively have a family of $2^{64}$~hash functions --- one for each possible seed value.

Given such a small family (i.e., given few random bits), we can prove that it must have high collision probabilities.
Indeed, consider the set of all strings of $m$~64-bit words. There are $2^{64 m}$ such strings. 
\begin{itemize}

\item Pick one hash function from the CityHash family. 
This function hashes every one of the $2^{64 m}$~strings to one of $2^{64}$~hash values. By a pigeonhole argument~\cite{Ros128}, there must be at least one hash value where  at least $2^{64 m}/2^{64}=2^{64 (m-1)}$~strings collide.
\item  Pick another hash function. Out of the $2^{64 (m-1)}$~strings colliding when using the first hash function, we must have $2^{64 (m-2)}$~strings also colliding when using the second hash function. 
\end{itemize}
We can repeat this process $m-1$~times  until we find  
$2^{64}$~strings  colliding when using any of these  
$m-1$~hash functions.  If an adversary picks  
any two of our $2^{64}$~strings and we pick the hash function at random in the whole family of $2^{64}$~hash functions, we get a collision with a probability of at least  $(m - 1)/2^{64}$.
 Thus, while we do not have a strict bound on the collision probability of the CityHash family, we know just from the small size of its seed that it must have a relatively high collision probability for long strings. 
In contrast,  \textsc{VHASH} and 
our
\textsc{CLHASH} 
(see \S~\ref{sec:clhash})
use more than 64~random bits and have correspondingly better collision bounds (see Table~\ref{table:comparison}).

\section{\textsc{VHASH}}

The \textsc{VHASH} family~\cite{dai2007vhash,vhashimpl} was designed
for 64-bit processors.  By default, it operates over 64-bit words. Among hash families offering good almost universality for large data inputs, \textsc{VHASH} might be the fastest 64-bit alternative on x64 processors --- except for our own proposal (see \S~\ref{sec:clhash}).  

 \textsc{VHASH} is $\epsilon$-almost $\Delta$-universal and builds on the 128-bit NH family~\cite{dai2007vhash}:
\begin{equation}
  \begin{aligned}
\label{eqn:nh}
\mathrm{NH}(s)=\sum_{i=1}^{l/2}\big (&(s_{2i-1}+k_{2i-1} \bmod 2^{64}) \\
&\times (s_{2i}+k_{2i} \bmod 2^{64})\big ) \bmod 2^{128}.
\end{aligned}
\end{equation} 
NH is $1/2^{64}$-almost $\Delta$-universal with hash values in $[0,2^{128})$. Although the NH family is defined only for inputs containing an even number of components, we can extend it to include
odd numbers of components by padding the input with a zero component.

 We can summarize \textsc{VHASH} (see Algorithm~\ref{alg:vhash}) as follows:
\begin{itemize}
\item 
 NH is used  to generate a 128-bit hash value for each block of 16~words. The result is $1/2^{64}$-almost $\Delta$-universal on each block.  
 \item These hash values are mapped to a value 
 in $[0,2^{126})$ by applying a modular reduction. 
These reduced hash values are then aggregated with a polynomial hash and finally reduced to a 64-bit value.
\end{itemize}
In total, the \textsc{VHASH} family is $1/2^{61}$-almost $\Delta$-universal over  $[0,2^{64}-257)$  for input strings of  up to $2^{62}$~bits~\cite[Theorem~1]{dai2007vhash}.

For long input strings, we expect that much of the running time of \textsc{VHASH} is in the computation of NH on blocks of 16~words. On recent x64 processors, this computation involves 8~multiplications using the \texttt{mulq} instruction (with two 64-bit inputs and two 64-bit outputs). For each group of two consecutive words ($s_i$ and $s_{i+1}$), we also need two 64-bit additions. To sum
all results, we need 7~128-bit additions that can be implemented using two 64-bit additions (\texttt{addq} and \texttt{adcq}). All of these operations have a throughput of at least 1~per cycle on  Haswell processors.
We can expect NH and, by extension, VHASH to be fast.

\textsc{VHASH} uses only 16~64-bit random integers for the NH family. As in \S~\ref{sec:hashingtuples}, we only need one specific NH function irrespective of the length of the string. \textsc{VHASH} also uses a 128-bit random integer $k$ and two more 64-bit random integers $k'_1$ and $k'_2$. Thus \textsc{VHASH} uses slightly less than 160~random bytes.

\begin{algorithm}
\caption{\textsc{VHASH} algorithm}\label{alg:vhash}
\begin{algorithmic}[1]\small
\REQUIRE 16 randomly picked 64-bit integers $k_1, k_2, \ldots, k_{16}$ defining a
 128-bit NH hash function (see Equation~\ref{eqn:nh}) over inputs of length 16
\REQUIRE $k$, a randomly picked element of $\{w 2^{96}+x 2^{64}+y 2^{32}+z \given \text{integers } w,x,y,z \in [0,2^{29})  \}$
\REQUIRE $k'_1, k'_2$, randomly picked integers in $[0,2^{64}-258]$
\STATE \textbf{input}: string $M$ made of $|M|$~bytes
\STATE Let $n$ be the number of 16-word blocks ($\lceil |M| / 16 \rceil$).
\STATE Let $M_i$   be the substring of $M$ from index $i$ to $i+16$, padding with zeros if needed.
\STATE Hash each $M_i$ using the NH function, labelling the result 128-bit results $a_i$ for $i=1,\ldots,n$.
\STATE Hash the resulting $a_i$ with a polynomial hash function and store the value in a 127-bit hash value $p$:  $p=k^n +a_1 k^{n-1}+\cdots + a_n + (|M| \bmod 1024) \times 2^{64} \bmod {(2^{127}-1)}$.\label{line:vhashfin}
\STATE Hash the 127-bit value $p$ down to a 64-bit value: $z=(p_1+k'_1) \times (p_2 +k'_2) \bmod {(2^{64}-257)}$, where $p_1 = p \div (2^{64}-2^{32)}$ and $p_2 = p \bmod  (2^{64}-2^{32})$.
\RETURN the 64-bit hash value $z$
\end{algorithmic}
\end{algorithm}

\subsection{Random Bits}

Nguyen and  Roscoe showed
that
at least $\log (m/\epsilon)$~random bits are required~\cite{Ros128},\footnote
{In the present paper, $\log n$ means $\log_2 n$.} where $m$ is the maximal string length in bits and $\epsilon$ is the collision bound.
For \textsc{VHASH}, the string length is limited to $2^{62}$~bits and the collision bound is $\epsilon = 1/2^{61}$. Therefore, 
for hash families offering the bounds of \textsc{VHASH}, we have that
$\log m/\epsilon = \log (2^{62} \times 2^{61})=123$~random bits are required.

That is,  16~random bytes are theoretically required to achieve the same collision bound as \textsc{VHASH} while many more are used (160~bytes)
This suggests that we might be able to find families using far fewer random bits while maintaining the same good bounds. 
In fact, it is not difficult to modify \textsc{VHASH} to reduce the use of random bits. It would
suffice to reduce the size of the blocks down from 16~words. We could
show that it cannot increase the bound on the collision probability by more than $1/2^{64}$.
However, reducing the size of the blocks has an adverse effect on speed. With large blocks and long strings,
most of the input is processed with the NH function before the more expensive polynomial hash
function is used. 
 Thus, there is  a trade-off between speed and the number of random bits, and \textsc{VHASH} is designed for speed on long strings.

\section{Finite Fields}

Our proposed hash family (\textsc{CLHASH}, see \S~\ref{sec:clhash}) works over a binary finite field.
For completeness, 
 we review field theory briefly, introducing (classical) results as needed for our purposes.

The real numbers form what is called a \emph{field}. A field is such that addition and multiplication are associative, commutative and distributive. We also have identity elements (0 for addition and 1 for multiplication). 
Crucially, all non-zero elements $a$ have an inverse  $a^{-1}$ 
(which is  
defined by $a \times a^{-1} = a^{-1} \times a  =1$). 

Finite fields (also called Galois fields) are fields containing a finite number of elements. All finite fields have cardinality $p^n$ for some prime $p$. Up to an \emph{algebraic isomorphism} (i.e., a one-to-one map preserving addition and multiplication), given a cardinality $p^n$, there is only one field (henceforth $GF(p^n)$). And for any power of a prime, there is a corresponding field.

\subsection{Finite Fields of Prime Cardinality}
\label{sec:finitefieldsprime}

It is easy to create  finite fields 
that have 
prime cardinality ($GF(p)$). 
Given $p$,  an instance of $GF(p)$ is  given by   the set of integers in $[0,p)$ with additions and multiplications completed by a modular reduction: 
\begin{itemize}
\item 
$a \times_{GF(p)} b \equiv a \times b \bmod p$
\item and  $a +_{GF(p)} b \equiv  a+b \bmod p$.
\end{itemize}
The numbers 0 and 1 are the identity elements. Given an element $a$, its additive inverse is $p-a$. 

It is not difficult to check that all non-zero elements have a multiplicative inverse. We review this classical result for completeness.  Given a non-zero element $a$ and two distinct $x, x'$, we have that $a x \bmod p \neq a x' \bmod p$ because $p$ is prime.  Hence, starting with a fixed non-zero element $a$, we have that the set  $\{a x \bmod p \, \mid \, x \in[0,p) \} $ has cardinality $p$ and must contain 1; thus, $a$ must have a multiplicative inverse.

\subsection{Hash Families in a Field}
\label{sec:hasfamexampleinfield}
Within a field, we can easily construct hash families having strong theoretical guarantees, as the next lemma illustrates.

\begin{lemma}\label{lemma:simpledeltauniversal}
The family of functions of the form 
 \begin{align*}h(x)=a x\end{align*}
in a finite field ($GF(p^n)$) is $\Delta$-universal, provided that the key $a$ is picked from all values of  the field. 
\end{lemma}

 As another example, consider hash functions of the form $h(x_1, x_2,\ldots, x_m)=a^{m-1} x_1 + a^{m-2} x_2+\cdots+ x_m$ where $a$ is picked at random (a \emph{random input}). 
Such \emph{polynomial hash functions} can be computed efficiently using Horner's rule: starting with $r=x_1$, compute $r\leftarrow a r + x_i$ for $i=2, \ldots, m$. Given any two distinct inputs, $x_1, x_2,\ldots, x_m$ and $x'_1, x'_2,\ldots, x'_m$, we have that $h(x_1, \ldots, x_m)-h(x'_1, \ldots, x'_m)$ is a non-zero polynomial of degree at most $m-1$ in $a$. By the fundamental theorem of algebra, we have that it is zero for at most $m-1$~distinct values of $a$. 
Thus we have that the probability of a collision is bounded by $(m-1)/p^n$ where $p^n$ is the cardinality of the field. For example, \textsc{VHASH} uses polynomial hashing with $p=2^{127}-1$
and $n=1$.

We can further reduce the collision probabilities if we use $m$~random inputs $a_1, \ldots, a_m$ picked  in the field to compute a \emph{multilinear} function: $h(x_1, \ldots, x_m)=a_1 x_1 + a_2 x_2 + \cdots + a_m x_m$.  We have $\Delta$-universality.  Given two distinct inputs, $x_1, \ldots, x_m$ and $x'_1, \ldots, x'_m$, we have that $x_i \neq x'_i$ for some $i$. Thus  we have that 
$h(x_1, \ldots, x_m)=c+h(x'_1, \ldots, x'_m)$  if and only if 
$a_i = (x_i - x'_i)^{-1} (c + \sum_{j\neq i} a_j (x'_j -x_j))$. 

 If $m$ is even, we can get the same bound on the collision probability with half the number of multiplications~\cite{703969,Lemire10072013,motzkin1955evaluation}: 
 \begin{align*}&h(x_1, x_2,\ldots, x_m)\\&=(a_1 + x_1) ( a_2 + x_2) + \cdots + (a_{m-1} + x_{m-1}) (a_m + x_m).\end{align*}
 The argument is similar.
 Consider that \begin{align*}& (x_i+a_i) (a_{i+1}+x_{i+1})
 -  (x'_i+a_i) (a_{i+1}+x'_{i+1})\\
& = a_{i+1} (x_i-x'_i) +a_i (x_{i+1}-x'_{i+1})+x_{i+1} x_i +x'_i x'_{i+1}.\end{align*}
 Take two distinct inputs, $x_1, x_2,\ldots, x_m$ and $x'_1, x'_2,\ldots, x'_m$.
 As before, we have that $x_i \neq x'_i$ for some $i$. Without loss of generality, assume that $i$ is odd; then we can 
find a unique solution
for $a_{i+1}$:
to do this, start 
 from $h(x_1, \ldots, x_m)=c+h(x'_1, \ldots, x'_m)$ and solve for  $a_{i+1} (x_i-x'_i)$ in terms of an expression that does not depend on $a_{i+1}$. Then use the fact that $x_i-x'_i$ has an inverse. 
 This shows that the collision probability is bounded by $1/p^n$ and we have $\Delta$-universality.

\begin{lemma}\label{lemma:halfisdeltauniversal}
Given an even number $m$, the family of functions of the form 
 \begin{align*}h(x_1, x_2,\ldots, x_m)=&(a_1 + x_1) ( a_2 + x_2) \\
  & + (a_3 + x_3) ( a_4 + x_4) \\
  & + \cdots \\
 & + (a_{m-1} + x_{m-1}) (a_m + x_m)\end{align*}
in a finite field ($GF(p^n)$) is $\Delta$-universal, providing that the keys $a_1, \ldots, a_m$ are picked from all values of  the field. 
In particular, the collision probability between two distinct inputs is bounded by $1/p^n$.
\end{lemma}

\subsection{Binary Finite Fields}
\label{sec:binaryfinitefields}

Finite fields having prime cardinality are simple (see \S~\ref{sec:finitefieldsprime}), but we would prefer to work with fields having a power-of-two cardinality (also called binary fields)  to match common computer architectures.
Specifically, we are interested in  $GF(2^{64})$ because our desktop processors
typically have
64-bit architectures.

We can implement such a field over the integers in $[0,2^L)$ by using the following two operations.
 Addition is defined as the bitwise XOR
 ($\oplus$) operation, which is fast on most computers: \begin{align*}a +_{GF(2^L)} b \equiv a \oplus b.\end{align*}
 The number 0 is the additive identity element ($a \oplus 0  = 0 \oplus a = a$), and
every number is its own additive inverse: $a \oplus a = 0$. Note that because binary finite fields use XOR as an addition, $\Delta$-universality and XOR-universality are effectively equivalent for our purposes in binary finite fields. 

 Multiplication is defined as a carry-less multiplication followed by a reduction. We use the convention that $a_i$ is the $i^{\mathrm{th}}$ least significant bit of integer $a$ and $a_i= 0$ if $i$ is larger than the most significant bit of $a$.  The $i^{\mathrm{th}}$ bit of the  carry-less multiplication $a \star b$ of $a$ and $b$ is given by
\begin{align}\label{eqn:carrylessmult}(a \star b)_i \equiv 
\bigoplus_{k=0}^i  a_{i-k} b_k\end{align}
where $a_{i-k} b_k$ is just a regular multiplication between two integers in $\{0,1\}$ and $\bigoplus_{k=0}^i$ is the bitwise XOR of a range of values.
The carry-less  product of two $L$-bit integers is a $2L$-bit integer.
We can check that the integers with $\oplus$ as addition and $\star$ as multiplication form a \emph{ring}: 
addition and multiplication are associative, commutative and distributive, and there is an additive identity element.
In this instance, the number 1 is a multiplicative identity element ($a \star 1 = 1 \star a = a$).
Except for the number 1, no number has a multiplicative inverse
in this ring.
 
Given the ring determined by $\oplus$ and $\star$, we can derive a corresponding finite field. However, just as with finite fields of prime cardinality, we need some kind of modular reduction and a concept equivalent to that of prime numbers\footnote{%
The general construction of a finite field of cardinality $p^n$ for $n>1$ is commonly explained in terms of polynomials with 
coefficients from $GF(p)$. To avoid unnecessary abstraction, we present finite fields of cardinality $2^L$ using regular $L$-bit integers. Interested readers can see Mullen and Panario~\cite{Mullen:2013:HFF:2555843}, for the 
alternative development.
}.

Let us define $\degree(x)$ to be the position of the most significant non-zero bit of $x$, starting at 0 (e.g., $\degree(1)=0$, $\degree(2)=1$, $\degree(2^j)=j$). For example, we have $\degree(x)\leq 127$ for any 128-bit integer $x$. Given any two non-zero integers $a,b$, we have that $\degree(a \star b) = \degree(a) + \degree (b)$ as a straightforward consequence of Equation~\ref{eqn:carrylessmult}. Similarly, we have that \begin{align*}\degree(a \oplus b) \leq \max (\degree(a),\degree(b)).\end{align*}

Not unlike regular multiplication, given integers $a,b$ with $b\neq 0$, there are unique integers $\alpha, \beta$ (henceforth the \emph{quotient} and the \emph{remainder}) such that  
\begin{align}a = \alpha \star b \; \oplus \; \beta \label{eqn:boring} \end{align}
where $\degree(\beta)<\degree(b)$.

 The uniqueness of the quotient and the remainder is easily shown. Suppose that there is another 
pair of values $\alpha',\beta'$ with the same property. Then $\alpha' \star b \; \oplus \; \beta' = \alpha \star b \; \oplus \; \beta$ which implies that  $(\alpha' \oplus \alpha) \star b = \beta' \oplus \beta$. However, since   $\degree(\beta' \oplus \beta)< \degree(b)$ we must have that $\alpha=\alpha'$.  From this it follows that $\beta=\beta'$, thus establishing uniqueness. 

We define $\div$ and $\bmod$ operators as giving respectively the quotient ($a\div b = \alpha$) and remainder ($a\mod b = \beta$) so that the equation 
\begin{align}a \equiv a\div b \star b \; \;\oplus \; \; a\bmod b\label{eqn:id}\end{align}
is an identity equivalent to Equation~\ref{eqn:boring}.  (To avoid unnecessary parentheses, we use the following operator precedence convention: $\star$,  $\bmod$ and $\div$ are executed first, from left to right, followed by $\oplus$.)

In the general case, we can compute $a \div b $ and $a \bmod b$ using a straightforward variation on 
the Euclidean division algorithm (see Algorithm~\ref{alg:division}) which proves the existence of the remainder and quotient. 
Checking the correctness of the algorithm is straightforward. We start initially with values $\alpha$ and $\beta$ such that $a=\alpha\star b \oplus \beta$.
By inspection, this equality is preserved throughout the algorithm. Meanwhile, the algorithm only terminates when the degree of $\beta$ is less than
that of 
 $b$, as required. And the algorithm must terminate, since the degree of $q$ is reduced by at least one each time it is updated (for a maximum of $\degree(a)-\degree(b)+1$~steps).
\begin{algorithm}
\caption{Carry-less division algorithm}\label{alg:division}
\begin{algorithmic}[1]\small
\STATE \textbf{input:} Two integers $a$ and $b$, where $b$ must be non-zero
\STATE \textbf{output:} Carry-less quotient and remainder: $\alpha = a \div b $ and $\beta = a \bmod b$, such that $a=\alpha \star b \oplus \beta$ and $\degree(\beta)<b$
\STATE Let $\alpha \leftarrow 0$ and $\beta \leftarrow a$  
\WHILE {$\degree(\beta)\geq \degree(b)$}
\STATE let $x\leftarrow 2^{\degree(\beta)-\degree(b)}$ 
\STATE $\alpha \leftarrow x\; \oplus \; \alpha$, 
$\beta \leftarrow x \star b \; \oplus \; \beta$ 
\ENDWHILE
\RETURN $\alpha$ and $\beta$
\end{algorithmic}
\end{algorithm}

Given $a = \alpha \star b \; \oplus \; \beta$ and $a' = \alpha' \star b \; \oplus \; \beta'$, we have that $a  \oplus a' = (\alpha \oplus \alpha') \star b \; \oplus \; (\beta \oplus \beta')$. Thus, 
it can be checked that  divisions and modular reductions are distributive:
\begin{align}\label{eqn:distributive}(a \oplus b) \bmod p =( a \bmod p ) \oplus ( b \bmod p),\end{align}
\begin{align}\label{eqn:distributivediv}(a \oplus b) \div p = (a \div p) \oplus (b \div p).\end{align}
Thus, we have $(a \oplus b) \bmod p = 0 \Rightarrow a \bmod p = b \bmod p $. 
Moreover, by inspection, we have that $\degree(a\bmod b)<\degree(b)$ and $\degree(a\div b)=\degree(a)-\degree(b)$.

The carry-less multiplication by a power of two is equivalent to  regular multiplication. For this reason,  a modular reduction by a power of two (e.g., $a \bmod{2^{64}}$) is just the regular integer modular reduction. Idem for division.

There are non-zero integers $a$ such that there is no integer $b$ other than 1 such that  $a\bmod b=0$; effectively $a$ is a prime number under the carry-less multiplication interpretation. These ``prime integers'' are more commonly known as \emph{irreducible polynomials} in the ring of polynomials $GF2[x]$, so we  call them \emph{irreducible} instead of prime. Let us pick such an irreducible  integer $p$ (arbitrarily) such that 
the degree of $p$ is 64. 
One such integer is
$2^{64}+2^{4}+2^{3}+2 +1$.
 Then we can finally define the multiplication operation in $GF(2^{64})$:
\begin{align*}a \times_{GF(2^{64})} b \equiv 
(a \star b) \bmod p.\end{align*}
Coupled with the addition $+_{GF(2^{64})}$ that is just a bitwise XOR, we have an implementation of the field  $GF(2^{64})$ over integers in $[0,2^{64})$. 

We call the index of the second most significant bit the \emph{subdegree}. We chose an irreducible $p$ of degree 64 having minimal subdegree (4).\footnote{This can be readily verified using a mathematical software package such as Sage or Maple.}   We  use the fact that this subdegree is small to accelerate the computation of the modular reduction in the next section.

\subsection{Efficient Reduction in $GF(2^{64})$}
\label{sec:efficient}
AMD and Intel have  introduced a fast instruction that can compute a carry-less multiplication between two 64-bit numbers, and it generates a 128-bit integer. To get the multiplication in $GF(2^{64})$, we must still reduce this 128-bit integer to a 64-bit integer. 
 Since there is no equivalent fast modular instruction,  we need to derive an efficient algorithm.  
 
 There are efficient reduction algorithms used in cryptography (e.g., from 256-bit to 128-bit integers~\cite{gueron2010efficient}), but they do not suit our purposes: we have to reduce to 64-bit integers.
 Inspired by the classical Barrett reduction~\cite{Barrett:1987:IRS:36664.36688},   
 Kne\v{z}evi\'{c} et al.\ proposed a generic modular reduction algorithm in $GF(2^n)$, using no more than two multiplications~\cite{springerlink:10.1007/978-3-540-69499-1_7}.
 We put this to good use in previous work~\cite{Lemire10072013}. However, we can do the same reduction using a single multiplication.
According to our tests, the  reduction technique presented next
is \SI{30}{\percent} faster than an optimized implementation based on  Kne\v{z}evi\'{c} et al.'s algorithm.

Let us write $p = 2^{64} \; \oplus \; r$. In our case, we have $r=2^{4}+2^{3}+2+1=27$ and $\degree(r)=4$. 
We are interested in applying a modular reduction by $p$ to  the result of the multiplication of two integers in $[0,2^{64})$, and the result of such a multiplication is an integer $x$ such that $\degree(x)\leq 127$. 
We want to compute $x \bmod p $ quickly. We begin with the following lemma.

\begin{lemma}\label{lemma:firstone}
Consider any 64-bit integer $p=2^{64}\oplus r$. 
We define the operations $\bmod$ and $\div$ as the counterparts of the carry-less multiplication $\star$ as in \S~\ref{sec:binaryfinitefields}.
Given any $x$, we have that
\begin{align*}&x \bmod p  \\ &= 
   ((z\div 2^{64}) \star 2^{64}) \bmod p \; \oplus \; z\bmod 2^{64}  \; \oplus \;  x\bmod 2^{64}
\end{align*}
where $z\equiv  (x\div 2^{64}) \star r$.
\end{lemma}
\begin{proof}
We have that
$x   = (x\div 2^{64}) \star 2^{64}   \;\oplus \; x\bmod 2^{64} 
 $
for any $x$ by definition. Applying the modular reduction on both sides of the equality, we get
\begin{align*}x \bmod p   &= (x\div 2^{64}) \star 2^{64} \bmod p  \; \oplus \;  x\bmod 2^{64} \bmod p &\\
 &= (x\div 2^{64}) \star 2^{64} \bmod p  \; \oplus \;  x\bmod 2^{64} \\ &\text{\textit{by Fact 1} }\\
 & = (x\div 2^{64}) \star r \bmod p  \; \oplus \;  x\bmod 2^{64} \\&\text{\textit{by Fact 2} }\\
 & = z \bmod p  \; \oplus \;  x\bmod 2^{64} \\ &\text{\textit{by z's def.} }\\
  & = 
   ((z\div 2^{64}) \star 2^{64}) \bmod p \; \oplus \; z\bmod 2^{64} \\& \; \oplus \;  x\bmod 2^{64} \\&\text{\textit{by Fact 3} }
 \end{align*}
 where Facts 1, 2 and 3 are as follows:
\begin{itemize}
\item (Fact 1) For any $x$, we have that $(x \bmod {2^{64}}) \bmod p = x \bmod {2^{64}}$. 
\item  (Fact 2) 
For any integer $z$, we have that  $(2^{64}\; \oplus \; r) \star z  \bmod p  = p \star z \bmod p = 0$ and therefore
\begin{align*}2^{64}\star z  \bmod p = r \star z  \bmod p\end{align*}
by the distributivity of the modular reduction (Equation~\ref{eqn:distributive}).
\item (Fact 3) Recall that by definition $z=(z\div 2^{64}) \star 2^{64}   \; \oplus \;  z\bmod 2^{64} $. We can substitute this equation in the equation from Fact 1. For any $z$ and any non-zero $p$, we have that 
\begin{align*}
z  \bmod p & =((z\div 2^{64}) \star 2^{64}   \; \oplus \;  z\bmod 2^{64}) \bmod p\\
& = ((z\div 2^{64}) \star 2^{64}) \bmod p \; \oplus \; z\bmod 2^{64}
\end{align*}
by the distributivity of the modular reduction (see Equation~\ref{eqn:distributive}).
\end{itemize}
 Hence the result is shown. \end{proof}

Lemma~\ref{lemma:firstone} provides a formula to compute $x\bmod p$.
Computing $z= (x\div 2^{64}) \star r$ involves a carry-less multiplication, which can be done efficiently on recent Intel and AMD processors. The computation of 
$z\bmod 2^{64}$  and $ x\bmod 2^{64}$ is trivial. It remains to compute 
$((z\div 2^{64}) \star 2^{64}) \bmod p $. At first glance, we still have a modular reduction. However, we can easily memoize the result of  
  $((z\div 2^{64}) \star 2^{64}) \bmod p$.
  The next lemma shows that there are only 16~distinct values to memoize (this
follows from the low subdegree of $p$).
  \begin{lemma}
  Given that $x$ has degree less than 128, 
  there are only  16~possible values of $(z\div 2^{64}) \star 2^{64} \bmod p$, where  $z\equiv (x\div 2^{64}) \star r$ and $r=2^{4}+2^{3}+2+1$.   \end{lemma}
\begin{proof}
Indeed,  we have that 
\begin{align*}\degree(z) = \degree(x) -64 + \degree(r).\end{align*} 
Because 
$\degree(x)\leq 127$, we have that  
 $\degree(z) \leq
127 - 64 + 4  
=67$. Therefore, we have $\degree(z \div 2^{64})\leq 3$. Hence, we can represent $z \div 2^{64}$ using 4~bits: there are only 16 4-bit integers.
\end{proof}

Thus, in the worst possible case, we would need to memoize 16~distinct 128-bit integers to represent $((z\div 2^{64}) \star 2^{64}) \bmod p $. However, observe that the degree of $z\div 2^{64}$ is bounded by $\degree(x)-64+4-64\leq 127-128+4=3$ since $\degree(x)\leq 127$. By using Lemma~\ref{lemma:notsobad}, we show that each integer $((z\div 2^{64}) \star 2^{64}) \bmod p $ has degree bounded by 7 so that it can be represented using no more than 
8~bits: setting $L=64$ and $w \equiv z\div 2^{64}$, $\degree(w)\leq 3$, $\degree(r)=4$ and $\degree(w)+\degree(r)\leq 7$.

Effectively, the lemma says that if you take a value of small degree $w$, you multiply it by $2^{{L}}$ and then compute the modular reduction on the result and a value $p$ that is almost $2^{L}$ (except for a value of small degree $r$), then the result has small degree: it is bounded by the sum of the degrees of $w$ and $r$. 

  \begin{lemma}\label{lemma:notsobad} 
  Consider
$p = 2^L \; \oplus \; r$,
with  $r$ of degree less than $L$.
For any $w$, the degree of 
$w \star 2^{L} \bmod p$ is bounded by $\degree(w) + \degree(r)$.

Moreover, when $\degree(w) + \degree(r)<L$ then the degree of 
$w \star 2^{L} \bmod p$ is exactly $\degree(w) + \degree(r)$.
  \end{lemma}
\begin{proof}
The result is trivial if  $\degree(w)+\degree(r)\geq L$, since the degree of $w \star 2^{L} \bmod p$ must be smaller than the degree of $p$. 

So let us assume that  $\degree(w)+\degree(r)<L$. By the definition of the modular reduction (Equation~\ref{eqn:id}), we have
\begin{align*}w  \star 2^{L} &= w  \star 2^{L} \div p \star p \; \;\oplus \; \; w  \star 2^{L}\bmod p.
\end{align*}
Let $w'=w  \star 2^{L} \div p$, then 
\begin{align*}
w  \star 2^{L} &= w' \star p \; \;\oplus \; \; w  \star 2^{L}\bmod p\\
 &= w' \star r  \; \;\oplus \; \; w' \star 2^{L} \; \;\oplus \; \; w  \star 2^{L}\bmod p.
 \end{align*}
The first $L$~bits of 
$w  \star 2^{L}$  and $w' \star 2^{L}$ are zero. Therefore, we have
\begin{align*}
( w' \star r  )\bmod 2^{L} &= (w  \star 2^{L}\bmod p) \bmod 2^{L}.
 \end{align*}
Moreover, the degree of $w'$ is the same as the degree of $w$: $\degree(w')=\degree(w)+\degree(2^{L})+\degree(p)=\degree(w)+L-L=\degree(w)$. Hence, we have  
 $\degree(w'\star r)=\degree(w)+\degree(r)<L$. And, of course, $\degree(w  \star 2^{L}\bmod p)<L$. Thus, we have that
\begin{align*}
w' \star r =  w  \star 2^{L}\bmod p.
\end{align*}
Hence, it follows that $\degree(w  \star 2^{L}\bmod p)= \degree(w' \star r)=\degree(w) +\degree(r)$.
\end{proof}

  Thus the memoization requires access to only 16~8-bit values.   We enumerate the values in question ($w \star 2^{64} \bmod p$ for $w = 0,1, \ldots, 15$) in Table~\ref{table:fact4}. 
 It is  convenient that  $16\times 8 =
128$~bits: 
the entire table fits in a 128-bit word. It means that if the list of 8-bit values are stored using one byte each, the SSSE3 instruction \texttt{pshufb} can be used for fast look-up. (See Algorithm~\ref{alg:modulo}.)

 \begin{table}[thb]
\caption{\label{table:fact4}Values of $w \star 2^{64} \bmod p$ for $w=0,1,\ldots, 15$ given $p=2^{64}+2^{4}+2^{3}+3$.}\centering \small
\begin{tabular}{cl|cl}
\toprule
  \multicolumn{2}{c}{$w$} & \multicolumn{2}{c}{$w \star 2^{64} \bmod p$} \\
 decimal & binary & decimal & binary \\\midrule
0   & $0000_2$& 0 & $00000000_2$\\
1 & $0001_2$&27&   $00011011_2$\\
2 & $0010_2$&54 &  $00110110_2$\\ 
3 & $0011_2$&45 &  $00101101_2$\\
4 & $0100_2$&108 & $01101100_2$\\
5 & $0101_2$&119&  $01110111_2$\\
6&  $0110_2$&90&   $01011010_2$\\
7 & $0111_2$&65&   $01000001_2$\\
8& $1000_2$&216&   $11011000_2$\\
9& $1001_2$&195&   $11000011_2$\\
10& $1010_2$&238&  $11101110_2$\\
11& $1011_2$&245&  $11110101_2$\\
12& $1100_2$&180&  $10110100_2$\\
13& $1101_2$&175&  $10101111_2$\\
14& $1110_2$&130&  $10000010_2$\\
15 & $1111_2$&153&  $10011001_2$\\\bottomrule
\end{tabular}
\end{table}

\begin{algorithm}
\caption{Carry-less division algorithm}\label{alg:modulo}
\begin{algorithmic}[1]\small
\STATE \textbf{input:} A 128-bit integer $a$ 
\STATE \textbf{output:} Carry-less modular reduction $a \bmod p$ where $p=2^{64}+27$
\STATE $z \leftarrow (a\div 2^{64}) \star r$
\STATE $w \leftarrow z\div 2^{64}$
\STATE Look up $w \star 2^{64} \bmod p$ in Table~\ref{table:fact4}, store result in $y$
\RETURN $a \bmod  2^{64} \oplus z \bmod 2^{64} \oplus y$ 
\end{algorithmic}
\centering\vspace{1em}
\begin{minipage}{0.9\textwidth}

Corresponding C implementation using x64  intrinsics:
\lstset{%
keywordstyle=\color{black}\bfseries\underbar,
commentstyle=\color{gray}, 
stringstyle=\ttfamily, 
showstringspaces=false,
morekeywords={uint64_t, __m128i,uint32_t},
language=C
}
\begin{lstlisting}
uint64_t  modulo( __m128i a) {
	__m128i r = _mm_cvtsi64_si128(27);
	__m128i z = 
	  _mm_clmulepi64_si128( a, r, 0x01);
	__m128i table = _mm_setr_epi8(0, 27, 54, 
	  45, 108,119, 90, 65, 216, 195, 238, 
	  245,180, 175, 130, 153);
	__m128i y = 
	  _mm_shuffle_epi8(table
	    ,_mm_srli_si128(z,8));
	__m128i temp1 = _mm_xor_si128(z,a);
	return _mm_cvtsi128_si64(
	  _mm_xor_si128(temp1,y));
}
\end{lstlisting}
\end{minipage}
\end{algorithm}

\section{\textsc{CLHASH}}
\label{sec:clhash}
The \textsc{CLHASH} family resembles the \textsc{VHASH} family --- except that
members work in a binary finite field. The \textsc{VHASH} family has the 128-bit NH family (see Equation~\ref{eqn:nh}), but we instead use the 128-bit CLNH family:
\begin{equation}
\label{eqn:clnh}
\mathrm{CLNH}(s)=
 \bigoplus_{i=1}^{l/2}\big ((s_{2i-1}\oplus k_{2i-1}) \star (s_{2i} \oplus k_{2i} )\big ) 
\end{equation} 
where the $s_i$ and $k_i$'s are 64-bit integers and $l$ is the length of the string $s$. The formula assumes that $l$ is even: we  pad odd-length inputs with a single zero word.
When an input string $M$ is made of $|M|$~bytes, we can consider it as string of 64-bit words $s$ by padding it with up to 7~zero bytes so that $|M|$ is divisible by 8.

On x64 processors with the CLMUL instruction set, a single term $ ((s_{2i-1}\oplus k_{2i-1}) \star (s_{2i} \oplus k_{2i} ))$ 
can be computed using one 128-bit XOR instructions (\texttt{pxor} in SSE2) and one carry-less multiplication using the \texttt{pclmulqdq} instruction:
\begin{itemize}
\item load $(k_{2i-1}, k_{2i})$ in a 128-bit word,
\item load $(s_{2i-1}, s_{2i})$ in another 128-bit word,
\item compute \begin{align*}(k_{2i-1}, k_{2i}) \oplus (s_{2i-1}, s_{2i}) \equiv (k_{2i-1} \oplus s_{2i-1}, k_{2i} \oplus s_{2i})\end{align*} using one \texttt{pxor} instruction,
\item compute $(k_{2i-1} \oplus s_{2i-1}) \star (k_{2i} \oplus s_{2i})$ using one \texttt{pcl\-mul\-qdq} instruction (result is a 128-bit word).
\end{itemize}
An additional \texttt{pxor} instruction is required per pair of words to compute CLNH, since we need to aggregate the results. 

We have that the family $s \to \mathrm{CLNH}(s) \bmod p$ for some irreducible $p$ of degree 64 is XOR universal over same-length strings. Indeed, $\Delta$-universality in the field $GF(2^{64})$ follows from Lemma~\ref{lemma:halfisdeltauniversal}. However, recall that  $\Delta$-universality in a binary finite field (with operations $\star$ and $\oplus$ for multiplication and addition) is the same as XOR universality --- addition is the XOR operation ($\oplus$). 
It follows that the CLNH family must be $1/2^{64}$-almost universal for same-length strings.

Given an arbitrarily long string of 64-bit words, we can divide it up into blocks of 128~words (padding
the last block
 with zeros if needed). Each block can be hashed using CLNH and the result is $1/2^{64}$-almost universal by Lemma~\ref{lemma:hashingtuples}.
If there is a single block, we can compute $\mathrm{CLNH}(s) \bmod p$ to get an XOR universal hash value. Otherwise, the resulting 128-bit hash values $a_1, a_2, \ldots, a_n$  can then be hashed once more.
For this we use
a polynomial hash function, $k^{n-1}  a_1 + k^{n-2}  a_2+\cdots+ a_n$, for some random input $k$ 
 in some finite field.
We choose the field $GF(2^{127})$ and use the irreducible $p=2^{127}+2+1$. We compute such a polynomial hash function by using 
Horner's rule: starting with $r=a_1$, compute $r\leftarrow k \star r \oplus  a_i$ for $i = 2, 3, \ldots , n$.  
For this purpose, we need carry-less multiplications between pairs of 128-bit integers: we can achieve the desired result with 4~\texttt{pclmulqdq} instructions, in addition to some shift and XOR operations.  The multiplication generates a 256-bit integer $x$ that must be reduced. However, it is not necessary to reduce it to
a 127-bit integer  (which would be the result if we applied a modular reduction by $2^{127}+2+1$). It is enough to reduce it to a 128-bit integer $x'$
such that $x' \bmod (2^{127}+2+1) = x \bmod (2^{127}+2+1)$. We get the desired result by setting $x'$ equal to the \emph{lazy} modular reduction~\cite{bluhm2013fast}
 $x \bmod_{\mathrm{lazy}} (2^{127}+2+1)$ defined as
\begin{equation}
  \begin{aligned} \label{eq:lazymod}
&  x \bmod_{\mathrm{lazy}} (2^{127}+2+1) \\  &\equiv  x \bmod ((2^{127}+2+1)\star 2) \\
& = x \bmod (2^{128}+4+2)\\
 &= (x \bmod 2^{128}) \oplus (( x \div 2^{128} ) \star 4 \oplus  ( x \div 2^{128} ) \star 2).
  \end{aligned}
\end{equation}
It is  computationally convenient to assume that $\degree(x)\leq 256-2$ so that $\degree(( x \div 2^{128} ) \star 4)\leq 128$. We can achieve this degree bound by picking the polynomial coefficient $k$ to have $\degree(k)\leq 128-2$.
The resulting polynomial hash family is $(n-1)/2^{126}$-almost universal for strings having the same length where $n$ is the number of 128-word blocks ($\lceil |M| / 1024 \rceil$ where $|M|$ is the string length in bytes), whether we use the actual modular or the lazy modular reduction.
 
 It remains to reduce the final output $\mathcal{O}$ (stored in a 128-bit word) to a 64-bit hash value. For this purpose, we can use $s \to \mathrm{CLNH}(s) \bmod p$ with $p=2^{64}+27$ (see \S~\ref{sec:efficient}),
and where $k''$ is a random 64-bit integer.
We treat
$\mathcal{O}$
as a string containing two 64-bit words. Once more,
the reduction 
is XOR universal by an application of Lemma~\ref{lemma:halfisdeltauniversal}.
  Thus, we have the composition of three hash functions with collision probabilities $1/2^{64}$, $(n-1)/2^{126}$ and $1/2^{64}$. It is reasonable to bound the string length by $2^{64}$~bytes: $n\leq 2^{64}/1024=2^{54}$. 
We have that $ 2/2^{64} + (2^{54}-1)/2^{126} < 2.004/2^{64}$. Thus, for same-length strings, we have 
$2.004/2^{64}$-almost XOR universality.

We further ensure that the result is XOR-universal over all strings:  $P\left (h(s) = h(s') \oplus  c \right ) \leq 1/2^{64}$ irrespective of whether $|s|= |s'|$. By Lemma~\ref{lemma:fromfixedtovariable},
 it suffices to XOR the hash value with $k'' \star |M| \bmod p$ where $k''$ is a random 64-bit integer  and $|M|$ is the string length as a 64-bit integer,  and where $p=2^{64}+27$. The XOR universality follows for strings having different lengths by Lemma~\ref{lemma:simpledeltauniversal} and the equivalence between XOR-universality and $\Delta$-universality in binary finite fields.
As a practical matter, since the final step involves the same modular reduction
twice in the expression 
 $(\mathrm{CLNH}(s)\bmod p) \oplus ((k'' \star |M|) \bmod p)$, we can simplify it  to
 $(\mathrm{CLNH}(s) \oplus (k'' \star |M|)) \bmod p$, thus avoiding an unnecessary modular reduction.

Our analysis is summarized by following lemma.
 
 \begin{lemma}
 \textsc{CLHASH} is $2.004/2^{64}$-almost XOR universal over strings of up to $2^{64}$~bytes.
 Moreover, it is XOR universal over strings of no more than \SI{1}{kB}. 
 \end{lemma}

The bound of the collision probability of \textsc{CLHASH} for long strings ($2.004/2^{64}$) is 4~times lower than  the  corresponding \textsc{VHASH} collision probability  ($1/2^{61}$).
For short strings (\SI{1}{kB} or less), \textsc{CLHASH} has a bound that is 8~times lower. 
See Table~\ref{table:comparison} for a comparison.
\textsc{CLHASH} is given by Algorithm~\ref{alg:clhash}.

\begin{table}
\caption{\label{table:comparison} Comparison between the two 64-bit hash families  \textsc{VHASH} and \textsc{CLHASH}}\centering
 \begin{tabular}{c|cc}\toprule
  & universality & input length \\\midrule 
 \textsc{VHASH} & $\frac{1}{2^{61}}$-almost $\Delta$-universal & 1--$2^{59}$~bytes \\[3ex]
  \multirow{2}{*}{\textsc{CLHASH}}
  &  XOR universal & 1--1024~bytes \\[1.5ex] 
  &    $\frac{2.004}{2^{64}}$-almost XOR universal & 1025--$2^{64}$~bytes \\ \\[1.5ex] \bottomrule
 \end{tabular}
\end{table}


\begin{algorithm}
\caption{\textsc{CLHASH} algorithm: all operations are carry-less, as per \S~\ref{sec:binaryfinitefields}. The $\gg$ operator indicates a left shift: $\mathcal{O} \gg 33$ is the value $\mathcal{O}$ divided by $2^{33}$.
}\label{alg:clhash}
\begin{algorithmic}[1]\small
\REQUIRE 128 randomly picked 64-bit integers $k_1, k_2, \ldots, k_{128}$ defining a
 128-bit CLNH hash function (see Equation~\ref{eqn:clnh}) over inputs of length 128
\REQUIRE $k$, a randomly picked 126-bit integer
\REQUIRE $k'$, a randomly picked 128-bit integer
\REQUIRE $k''$, a randomly picked 64-bit integer
\STATE \textbf{input}: string $M$ made of $|M|$~bytes
\IF {$|M|\leq 1024$}
\STATE $\mathcal{O}\leftarrow \mathrm{CLNH}(M) \oplus (k'' \star |M|) \bmod (2^{64}+27) $  
\RETURN  $\mathcal{O} $
\ELSE
\STATE Let $n$ be the number of 128-word blocks ($\lceil |M| / 1024 \rceil$).
\STATE Let $M_i$   be the substring of $M$ from index $128i$ to $128 i+127$ inclusively, padding with zeros if needed. \label{line:thirteen}
\STATE Hash each $M_i$ using the CLNH function, labelling the result 128-bit results $a_i$ for $i=1,\ldots,n$.
That is, $a_i \leftarrow \mathrm{CLNH}(M_i)$.
\STATE Hash the resulting $a_i$ with a polynomial hash function and store the value in a 128-bit hash value $\mathcal{O}$:  $\mathcal{O}\leftarrow a_1 \star k^{n-1}\oplus \cdots \oplus  a_n  \bmod_{\mathrm{lazy}} (2^{127}+2+1) $ (see Equation~\ref{eq:lazymod}).
\STATE Hash the 128-bit value $\mathcal{O}$, treating it as two 64-bit words ($\mathcal{O}_1, \mathcal{O}_2$), down to a 64-bit  CLNH hash value (with the addition of a term accounting for the length $|M|$ in bytes)
\begin{align*}z\leftarrow ((\mathcal{O}_1\oplus k'_1) \star (\mathcal{O}_2 \oplus k'_2 ) \oplus (k'' \star |M|) \bmod{ (2^{64}+27)}.\end{align*}
Values $k'_1$ and $k'_2$ are the two 64-bit words contained in $k'$.
\RETURN the 64-bit hash value $z $
\ENDIF
\end{algorithmic}
\end{algorithm}

\subsection{Random Bits}

One might wonder whether using \SI{1}{kB} of random bits
is necessary.
For strings of no more than \SI{1}{kB}, \textsc{CLHASH} is XOR universal. Stinson showed
that in such cases, we need the number of random bits to match the input length~\cite{stinson1996connections}. 
That is, we need at least \SI{1}{kB} to achieve XOR universality over strings having \SI{1}{kB}.
Hence,  \textsc{CLHASH} makes nearly optimal use of the random bits.

\section{Statistical Validation}
\label{sec:statvalid}
Classically, hash functions have been deterministic: fixed maps $h$ from $U$ to $V$, where $|U| \gg |V|$ and
thus collisions are inevitable.
Hash functions might be assessed according to whether their
outputs are distributed evenly, i.e., whether $|h^{-1}(x)| \approx |h^{-1}(y)|$ for two
distinct $x,y \in V$.
However, in practice, the actual input is likely to
consist of clusters of nearly identical keys~\cite{knuth:v3}: for instance,  symbol table
entries such as \texttt{temp1}, \texttt{temp2}, \texttt{temp3} are to be expected,
or a collection of measured data values is likely to contain clusters of similar
numeric values.  Appending an extra character to the end of an input string, or
flipping a bit in an input number, should (usually) result in a different hash value.
A collection of desirable properties can be defined, and then hash functions rated on their
performance on data that is meant to represent realistic cases. 

One common use of randomized hashing is to avoid 
DoS (denial-of-service) 
attacks when an
adversary controls the series of keys submitted to a hash table.  In this setting,
prior to the use of a hash table, a random selection of hash function is made from
the family.  The (deterministic) function is then used, at least until the number
of collisions is observed to be too high. A high number of collisions presumably indicates the hash table
needs to be resized, although it could indicate that an undesirable member of the
family had been chosen.    Those contemplating switching from deterministic hash
tables to randomized hash tables would like to know that typical performance would
not degrade much.  Yet, as carefully tuned deterministic functions can sometimes outperform
random assignments for typical inputs~\cite{knuth:v3}, some degradation might need 
to be tolerated.  Thus, it is worth checking a few randomly chosen members of
our \textsc{CLHASH}  families against statistical tests.

\subsection{SMHasher}

The SMHasher program~\cite{smhasher} includes a variety of quality tests on
a number of minimally randomized hashing algorithms, for which we have weak or
no known theoretical guarantees.  It runs several statistical tests, such as the following.
\begin{itemize}
\item Given a randomly generated input, changing a few bits at random should not generate a collision.
\item Among all inputs containing only two non-zero bytes (and having a fixed length in $[4,20]$), collisions should be unlikely (called the \emph{TwoBytes} test).
\item  Changing a single bit in the input should change  half the bits of the hash value, on average~\cite{estebanez2006evolving} (sometimes called the \emph{avalanche effect}). 
\end{itemize}
Some of these tests are demanding:  e.g., CityHash~\cite{cityhash} fails the \emph{TwoBytes} test.

 
 We added both \textsc{VHASH} and \textsc{CLHASH} to SMHasher and
 used the
Mersenne Twister (i.e., MT19937) to generate the random bits~\cite{matsumoto1998mtd}.
 We find that \textsc{VHASH}  passes all tests. However,
 \textsc{CLHASH} fails one of them: the avalanche test.
 We can illustrate the failure.
Consider that for short fixed-length strings (8~bytes or less),  \textsc{CLHASH} is effectively equivalent to a hash function of the form $h(x) = a 
 \star x \bmod p$, where $p$ is irreducible. Such hash functions form an XOR universal family. They also satisfy the identity  $h(x \oplus y) \oplus h(x) = h(y)$. It follows that no matter what value $x$ takes, modifying the same $i^{\mathrm{th}}$~bit modifies the resulting hash value in a consistent manner (according to $h(2^{i+1})$). We can still expect that changing a bit in the input changes  half the bits of the hash value on average. However, SMHasher checks that $h(x \oplus 2^{i+1})$ differs from $h(x)$ in any given bit about half the time over  many randomly chosen inputs $x$. Since $h(x \oplus 2^{i+1})\oplus h(x)$ is independent from $x$ for short inputs with  \textsc{CLHASH}, any given bit is either always flipped (for all $x$) or never. Hence,  \textsc{CLHASH} fails the SMHasher test.  

Thankfully, we can slightly modify CLHASH so that all tests pass if we so desire.
It suffices to apply an  additional \emph{bit mixing} function taken from  MurmurHash~\cite{smhasher} to the result of \textsc{CLHASH}. The function consists of two multiplications and three shifts over 64-bit integers:
\begin{align*}
x & \leftarrow   x \oplus (x \gg 33),\\
x & \leftarrow   x \times 18397679294719823053,\\
x & \leftarrow   x \oplus (x \gg 33),\\
x & \leftarrow   x \times 14181476777654086739,\\
x & \leftarrow   x \oplus (x \gg 33).
\end{align*}
Each step is a bijection: e.g., multiplication
 by an odd integer is always invertible. A bijection  does not affect collision bounds.

\section{Speed Experiments}

We implemented a performance benchmark in C and compiled our software using GNU GCC~4.8 with  the \texttt{-O2} flag. 
The benchmark program ran on  a Linux server with an Intel i7-4770 processor running at \SI{3.4}{GHz}.
This CPU  has \SI{32}{kB} of L1 cache, \SI{256}{kB} of L2 cache per core,
and \SI{8}{MB} of L3 cache shared by all cores.
The machine has \SI{32}{GB} of RAM (DDR3-1600 with double-channel). We disabled Turbo Boost and set the processor
to run only
at its highest clock speed, effectively disabling the processor's power management. 
All timings are done using the time-stamp coun\-ter (\texttt{rdtsc})~instruction~\cite{intelbenchmark}.
Although all our software\footnote{Our
benchmark software
is made
freely available
under a liberal open-source license (\url{https://github.com/lemire/StronglyUniversalStringHashing}), 
and it includes the
modified SMHasher as well as all the necessary software to reproduce our results.}
is single-threaded, we disabled hyper-threading as well.

Our experiments compare implementations of \textsc{CLHASH}, \textsc{VHASH}, SipHash~\cite{siphash2012}, GHASH~\cite{gueron2010efficient} and Google's CityHash.
\begin{itemize}
\item
We implemented \textsc{CLHASH} using Intel intrinsics.
As described in \S~\ref{sec:clhash}, we use various \emph{single instruction, 
multiple data} (SIMD) instructions (e.g., SSE2, SSE3 and SSSE3) in addition to the CLMUL instruction set.  
The random bits are stored consecutively in memory, aligned with a cache line (64~bytes). 
\item
For \textsc{VHASH}, we used the authors' 64-bit implementation~\cite{vhashimpl}, which is optimized with inline assembly. It stores the random bits 
in a C \texttt{struct}, and we do not include the overhead of constructing this \texttt{struct} in the timings. The authors assume that the input length is divisible by 16~bytes, or padded with zeros to the nearest 16-byte boundary. In some instances, we would need to copy part of the input to a new location prior to hashing the content to satisfy the requirement. Instead, we decided to optimistically hash the data in-place without copy. Thus, we slightly overestimate the speed of the  \textsc{VHASH} implementation --- especially on shorter strings. 
\item We used the reference C implementation of SipHash~\cite{siphash}. SipHash is a fast family of 64-bit pseudorandom hash functions adopted, among others, by the Python language.
\item CityHash is commonly used in applications where high speed is desirable~\cite{Lim:2014:MHA:2616448.2616488,Fan:2014:CFP:2674005.2674994}.
We wrote a simple C port of Google's  CityHash (version~1.1.1)~\cite{cityhash}. 
Specifically, we benchmarked the \texttt{CityHash64WithSeed} function.  
\item Using Gueron and Kounavis'~\cite{gueron2010efficient}
code, we implemented a fast version of \textsc{GHASH} accelerated with the CLMUL instruction set. \textsc{GHASH}  is a polynomial hash function over $GF(2^{128})$ using the irreducible polynomial $x^{128}+ x^7+x^2+x+1$: $h(x_1,x_2, \ldots, x_n)= a^n x_1 + a^{n-1} x_2+\ldots + a x_n$ for some 128-bit key $a$. To accelerate computations, Gueron and Kounavis  replace the traditional Horner's rule  with an extended version that processes input words four at a time: starting with $r=0$ and precomputed powers $a^2, a^3, a^4$, compute   $r\leftarrow a^4 (r+ x_i) + a^3 x_{i+1} + a^2 x_{i+2} + a  x_{i+3} $ for $i=1, 4, \ldots, 4 \lfloor m/4 \rfloor -3$. We complete the computation with the usual Horner's rule when the number of input words is not divisible by four. In contrast with other hash functions, \textsc{GHASH} generates 128-bit hash values. 
\end{itemize}

\textsc{VHASH}, \textsc{CLHASH} and \textsc{GHASH} require random bits. The time spent by the random-number generator is excluded from the timings.

\subsection{Results}

We find that the hashing speed is not sensitive to the content of the inputs --- thus we generated the inputs using a random-number generator.
For any given input length, we repeatly hash the strings so that, in total, 40~million input words have been processed.

As a first test, we hashed \SI{64}{B}~and \SI{4}{kB}~inputs (see Table~\ref{table:results3}) and we report the number of cycles spent to hash one byte: for \SI{4}{kB}~inputs, we got 0.26 for \textsc{VHASH},\footnote{For comparison, Dai and Krovetz
reported that VHASH
used 0.6~cycles per byte on
an Intel Core 2 processor (Merom)~\cite{vhashimpl}.}
 0.16 for \textsc{CLHASH}, 0.23 for \textsc{CityHash} and 0.93 for \textsc{GHASH}. That is,  \textsc{CLHASH} is over \SI{60}{\percent} faster than \textsc{VHASH} and almost \SI{45}{\percent} faster than CityHash.  
Moreover, 
SipHash is an order of magnitude slower.  Considering that it produces 128-bit hash values, the PCMUL-accelerated \textsc{GHASH} offers good performance: it uses less than one cycle per input byte for long inputs.

Of course, the relative speeds depend on the length of the input. In Fig.~\ref{fig:comp}, we vary the input length from 8~bytes to \SI{8}{kB}. We see that the results for input lengths of \SI{4}{kB} are representative.  Mostly, we have that \textsc{CLHASH} is \SI{60}{\percent} faster than \textsc{VHASH} and \SI{40}{\percent} faster than CityHash. However, CityHash  and \textsc{CLHASH} have similar performance for  small inputs (32~bytes or less) whereas \textsc{VHASH} fares poorly over these same small inputs. We find that SipHash is not competitive in these tests. 

\begin{table}
\caption{\label{table:results3}A comparison of estimated CPU cycles per byte on a Haswell Intel processor using \SI{4}{kB}~inputs. All schemes generate 64-bit hash values, except that GHASH generates 128-bit hash values. 
}
\centering
\begin{tabular}{ccc}\hline
scheme                   & \SI{64}{B} input & \SI{4}{kB} input \\ \hline
 \textsc{VHASH} &  1.0 & 0.26 \\
 \textsc{CLHASH} & \textbf{ 0.45} & \textbf{0.16} \\ 
CityHash & 0.48  & 0.23
\\
SipHash & 3.1 & 2.1
\\
GHASH & 2.3 & 0.93
\\
\hline
\end{tabular}
\end{table}

\begin{figure*}\centering
\subfloat[Cycles per input byte]{%
\includegraphics[width=0.49\textwidth]{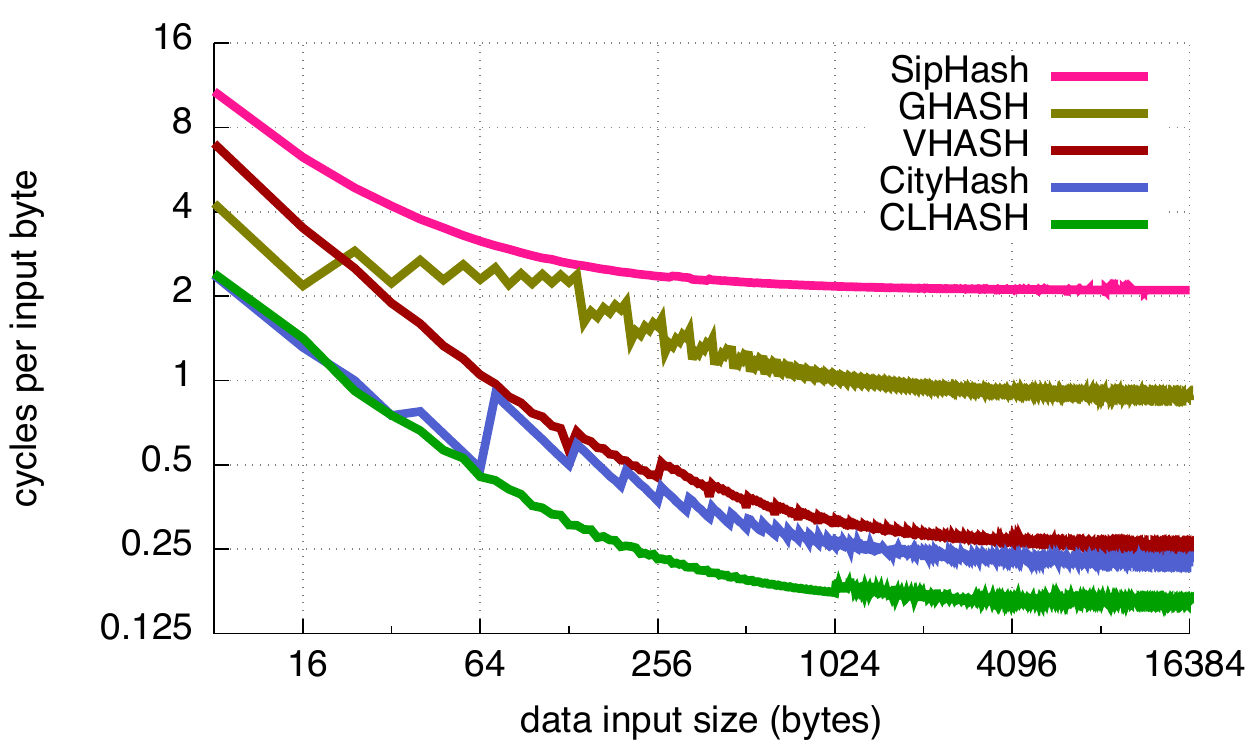}
}
\subfloat[Ratios of cycles per input byte]{%
\includegraphics[width=0.49\textwidth]{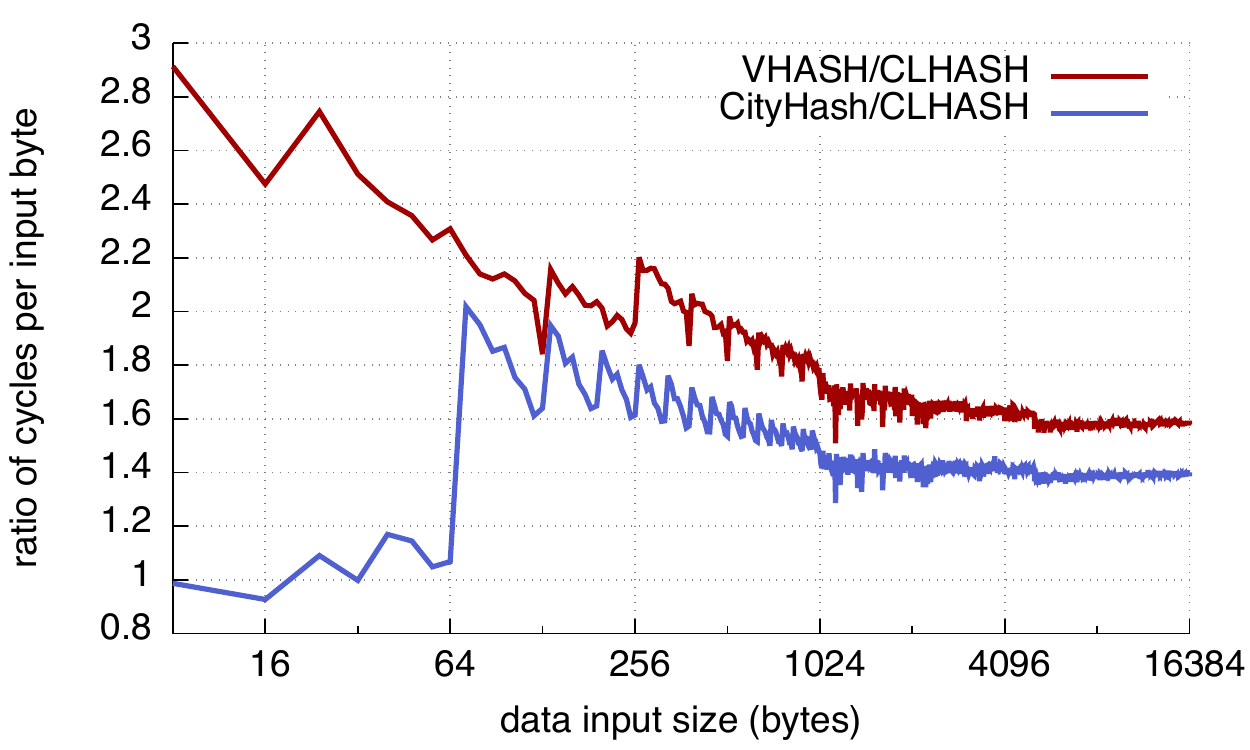}
}
\caption{\label{fig:comp}Performance comparison for various input lengths. For large inputs, \textsc{CLHASH} is faster, followed in order of decreasing speed by CityHash, \textsc{VHASH}, \textsc{GHASH} and SipHash.
}
\end{figure*}

\subsection{Analysis}
From an algorithmic point of view, \textsc{VHASH} and \textsc{CLHASH} are similar. Moreover, \textsc{VHASH} uses a conventional multiplication operation that has lower latency and higher throughput than \textsc{CLHASH}.
And the \textsc{VHASH} implementation relies on hand-tuned assembly code.
 Yet \textsc{CLHASH} is \SI{60}{\percent} faster.

For long strings, the bulk of the \textsc{VHASH} computation is spent computing the NH function.
When computing NH,  
each pair of input words (or 16~bytes) uses the following instructions: one \texttt{mulq}, three \texttt{adds} and one \texttt{adc}.
Both \texttt{mulq} and \texttt{adc} generate two micro-operations (\microops{}) each, so without counting register loading operations, we need at least $3+2\times 2 =7$~\microops{} to process two words~\cite{fog2014instruction}.
 Yet Haswell processors, like other recent Intel processors,
are apparently limited to a sustained execution of no more than
\SI{4}{\microops{}} per cycle. 
Thus we need at least $7/4$~cycles for every 16~bytes. That is, \textsc{VHASH} needs at least 0.11~cycles per byte.
 Because \textsc{CLHASH} runs at 0.16~cycles per byte on long strings (see Table~\ref{table:results3}),
 we have that no implementation of  \textsc{VHASH} could surpass our implementation of \textsc{CLHASH} by more than
 \SI{35}{\percent}.
 Simply put, \textsc{VHASH} requires too many \microops{}.

\textsc{CLHASH} is not similarly limited.
For each pair of input 64-bit words, CLNH uses two 128-bit XOR instructions (\texttt{pxor}) and one \texttt{pclmulqdq} instruction. Each \texttt{pxor} uses one (fused) \microop{} whereas the \texttt{pclmulqdq} instruction  uses two 
 \microops{} for a total of 
 \SI{4}{\microops{}}, versus the \SI{7}{\microops{}} 
absolutely needed by \textsc{VHASH}.
Thus, the number of \microops{} dispatched per cycle is less likely to be a bottleneck for \textsc{CLHASH}. However, the \texttt{pclmulqdq} instruction has a throughput of only two~cycles per instruction.
Thus, we can only process one pair of  64-bit words every two cycles, for a speed of $2/16=0.125$~cycles per byte. The measured speed (0.16~cycles per byte) is about \SI{35}{\percent} higher than this lower bound according to Table~\ref{table:results3}. This suggests that our implementation of \textsc{CLHASH} is nearly optimal --- at least for long strings.
 We verified our analysis with the IACA code analyser~\cite{intelIACA}. It reports that \textsc{VHASH} is indeed limited by the number of \microops{}  that can be dispatched per cycle, unlike \textsc{CLHASH}.

\section{Related Work}

The work that lead to the design of the \texttt{pclmulqdq} instruction by Gueron and Kounavis~\cite{gueron2010efficient} introduced efficient algorithms using this instruction, e.g., an algorithm for 128-bit modular reduction in  Galois Counter Mode. Since then, 
the \texttt{pclmulqdq} instruction has been used to speed up 
cryptographic applications.
Su and Fan find that the Karatsuba formula becomes especially efficient for software implementations of multiplication in binary finite fields due to the \texttt{pclmulqdq}  instruction~\cite{su2012impact}. Bos et al.~\cite{bos2011efficient} used the CLMUL instruction set for 256-bit hash functions on the Westmere microarchitecture.
Elliptic curve cryptography benefits from the \texttt{pclmulqdq} instruction~\cite{oliveirafast,Oliveira2014,Taverne2011}. 
Bluhm and Gueron pointed out
that the benefits are increased on the Haswell microarchitecture due to the higher throughput and lower latency of the instruction~\cite{bluhm2013fast}.

In previous work, we used the \texttt{pclmulqdq} instruction for fast 32-bit random hashing on the Sandy Bridge and Bulldozer architectures~\cite{Lemire10072013}. However, our results were disappointing, due in part to the low throughput of the instruction on these older microarchitectures.

\section{Conclusion}

The \texttt{pclmulqdq} instruction on recent Intel processors 
enables a fast and almost universal 64-bit hashing family 
(\textsc{CL\-HASH}). In terms of raw speed, the hash functions
from this family can surpass some of the  fastest 64-bit hash functions on x64 processors
(\textsc{VHASH} and CityHash). Moreover, \textsc{CLHASH} offers
superior bounds on the collision probability. 
\textsc{CLHASH} makes optimal use of the random bits, in the sense that it 
offers XOR universality for short strings (less than \SI{1}{kB}).

We believe that \textsc{CLHASH} might be suitable for many common purposes.
The \textsc{VHASH} family has been proposed for cryptographic
applications, and specifically  message authentication (VMAC): similar applications are possible for \textsc{CLHASH}. Future work should investigate these applications.

Other microprocessor architectures also support fast carry-less multiplication, sometimes referring to it as \emph{polynomial multiplication} (e.g., ARM~\cite{arm8} and Power~\cite{power2013}). Future work might review the performance of \textsc{CLHASH} on these architectures.
It might also consider the acceleration of alternative hash
families such as those based on Toeplitz matrices~\cite{stinson1996connections}.

\begin{acknowledgements}
This work is supported by the National Research Council of Canada, under grant 26143.
\end{acknowledgements}
\bibliographystyle{spmpsci}
\bibliography{clmul} 

\end{document}